\documentclass[12pt]{amsart}
\usepackage{amssymb,color}
\numberwithin{equation}{section} 

\usepackage{graphicx}

\newcommand{\bea}{\begin{eqnarray}}
\newcommand{\eea}{\end{eqnarray}}
\newcommand{\ba}{\begin{array}}
\newcommand{\ea}{\end{array}}
\newcommand{\edc}{\end{document}}
\newcommand{\bc}{\begin{center}}
\newcommand{\ec}{\end{center}}
\newcommand{\be}{\begin{equation}}
\newcommand{\ee}{\end{equation}}

\newcommand{\dsf}{\displaystyle\frac}

\def\cb{{\mathcal B}}

\def\cf{{\mathcal F}}
\def\cg{{\mathcal G}}
\def\ch{{\mathcal H}}

\def\bc{{\mathbb C}}

\def\bn{{\mathbb N}}

\def\bq{{\mathbb Q}}
\def\br{{\mathbb R}}

\def\bz{{\mathbb Z}}

  \def\G{\Gamma}
\def\d{\delta}  

\def\e{\epsilon}
\def\l{\lambda} 

\def\m{\mu}

\def\n{\nu}

\def\s{\sigma} 
\def\t{\theta}

\def\w{\omega} \def\Om{\Omega}

\def\h{{\mathbf{h}}}

\def\xb{{\mathbf{x}}}
\def\sb{{\mathbf{s}}}
\def\yb{{\mathbf{y}}}
\def\Bb{{\mathbf{B}}}

\newtheorem{thm}{Theorem}[section]
\newtheorem{lem}[thm]{Lemma}

\theoremstyle{remark}
\newtheorem{rem}{Remark}[section]
\newtheorem{obs}{Observation}[section]

\begin{document}

\title[$p$-adic quasi Gibbs measure]
{Existence of $P$-adic quasi Gibbs measure for countable state Potts
model on the Cayley tree}


\author{Farrukh Mukhamedov}
\address{Farrukh Mukhamedov\\
 Department of Computational \& Theoretical Sciences\\
Faculty of Science, International Islamic University Malaysia\\
P.O. Box, 141, 25710, Kuantan\\
Pahang, Malaysia} \email{{\tt far75m@yandex.ru}, {\tt
farrukh\_m@iium.edu.my}}

\begin{abstract}
In the present paper we provide a new construction of measure,
called $p$-adic quasi Gibbs measure, for countable state of $p$-adic
Potts model on the Cayley tree. Such a construction depends on a
parameter $\frak{p}$ and wights. In particular case, i.e. if
$\frak{p}=\exp_p$, the defined measure coincides with $p$-adic Gibbs
measure. In this paper, under some condition on weights we establish
the existence of $p$-adic quasi Gibbs measures associated with the
model. Note that this condition does not depend on values of the
prime $p$. An analogues fact is not valid when the number of spins
is finite. \vskip 0.3cm \noindent {\it
Mathematics Subject Classification}: 46S10, 82B26, 12J12.
\end{abstract}

\maketitle

\section{introduction}

Interest in the physics of non-Archimedean quantum models
\cite{ADFV,BC,FW,MP,VVZ} is based on the idea that the structure of
space-time for very short distances might conveniently be described
in terms of non-Archimedean numbers. One of the ways to describe
this violation of the Archimedean axiom, is the using $p$-adic
analysis. Numerous applications of this analysis to mathematical
physics have been proposed in
\cite{ABK},\cite{Kh1},\cite{Kh2},\cite{MP}. It is known \cite{Kh2}
that a number of $p$-adic models in physics cannot be described
using ordinary Kolmogorov's probability theory. New probability
models - $p$-adic probability models were investigated in
\cite{K3},\cite{KhN},\cite{KYR}. In \cite{KL,Lu,K4} the theory of
stochastic processes with values in $p$-adic and more general
non-Archimedean fields having probability distributions with
non-Archimedean values has been developed. This gives a possibility
to develop the theory of statistical mechanics in the context of the
$p$-adic theory, since it lies on the base of the theory of
probability and stochastic processes. The central problems of that
theory \cite{G} is the study of infinite-volume Gibbs measures
corresponding to a given Hamiltonian, which also includes a
description of the set of Gibbs measures. In most cases such
analysis depend on a specific properties of Hamiltonian, and
complete description is often a difficult problem. This problem, in
particular, relates to a phase transitions problem.

In \cite{KK1,KK2} a notion of ultrametric Markovianity, which
describes independence of contributions to random field from
different ultrametric balls has been introduced, and showed that
Gaussian random fields on general ultrametric spaces (which were
related with hierarchical trees), which were defined as a solution
of pseudodifferential stochastic equation, satisfy the Markovianity.
Some applications of the results to replica matrices, related to
general ultrametric spaces have been investigated in \cite{KK3}( see
also \cite{DF}).

In \cite{MR1,MR2} we have developed of $p$-adic probability theory
approaches to study $q+1$-state nearest-neighbor $p$-adic Potts
model on Cayley tree (see \cite{W}). We constructed infinite volume
$p$-adic Gibbs measures for the mentioned model, and moreover,
established the existence of a phase transition (here the phase
transition means the existence of two distinct $p$-adic Gibbs
measures for the given model). Further, in \cite{Mq} we have
introduced a new kind of $p$-adic measures for the mentioned model,
called {\it $p$-adic quasi Gibbs measure}. For such a model, we
investigated a phase transition phenomena from the associated
dynamical system point of view. Namely, we established that if $q$
is not divisible by $p$, then there occurs the quasi phase
transition. Note that such kind of measures present more natural
concrete examples of $p$-adic Markov processes (see \cite{KL}, for
definitions). In \cite{KM,KM1,M} we investigated a countable state
$p$-adic Potts model on the Cayley tree\footnote{The classical (real
value) counterparts of such models were considered in \cite{G,GR}.},
and provided a construction of $p$-adic Gibbs measures which depends
on weights $\l$. Moreover, the uniqueness of such measures under
certain conditions to the weight was proved.

In this paper we continue our investigations on countable state
$p$-adic Potts model. Namely, we are going to provide a new
construction of measure, called $p$-adic quasi Gibbs measure, for
the mentioned model on the Cayley tree. Such a construction depends
on a parameter $\frak{p}$ and wights $\{\l(i)\}_{i\in\bn}$. In
particular case, i.e. if $\frak{p}=\exp_p$, the defined measure
coincides with $p$-adic Gibbs measure (see \cite{KM,KM1}). In this
paper, under some condition on weights we establish the existence of
$p$-adic quasi Gibbs measures associated with the model. Note that
this condition does not depend on values of the prime $p$. An
analogues fact is not valid when the number of spins is finite.
Moreover, our result extends the previous proved ones in
\cite{KM1,M}. Note that in comparison to a real case, in a $p$-adic
setting, \`{a} priori the existence of such kind of measures for the
model is not known, since there is not much information on
topological properties of the set of all $p$-adic measures defined
even on compact spaces. However, in the real case, there is the so
called the Dobrushin's Theorem \cite{Dob1,Dob2} which gives a
sufficient condition for the existence of the Gibbs measure for a
large class of Hamiltonians. Note that when states are finite, then
the corresponding $p$-adic Potts models on the same trees have been
studied in \cite{Mq}.

\section{Preliminaries}

Fix a prime number $p$, which throughout the paper will be a fixed
greater than 3, and let $\bq_p$ denote the field of $p$-adic
filed, formed by completing $\bq$  with respect to the unique
absolute value satisfying $|p| = 1/p$. The absolute value
$|\cdot|$, is non- Archimedean, meaning that it satisfies the
ultrametric triangle inequality $|x + y|_p \leq \max\{|x|_p,
|y|_p\}$.

Given $a\in \bq_p$ and  $r>0$ put
$$
B(a,r)=\{x\in \bq_p : |x-a|_p< r\}, \ \ S(a,r)=\{x\in \bq_p :
|x-a|_p=r\}.
$$
The {\it $p$-adic exponential} is defined by
$$
\exp_p(x)=\sum_{n=0}^{\infty}\dsf{x^n}{n!},
$$
which converges for $x\in B(0,p^{-1/(p-1)})$.

\begin{lem}\label{pr}\cite{KM} If $|a_i|_p\leq 1$, $|b_i|_p\leq 1$, $i=1,\dots,n$, then
\begin{equation*}
\bigg|\prod_{i=1}^{n}a_i-\prod_{i=1}^n b_i\bigg|_p\leq \max_{i\leq
i\leq n}\{|a_i-b_i|_p\}
\end{equation*}
\end{lem}

Note the basics of $p$-adic analysis, $p$-adic mathematical physics
are explained in \cite{Ko,VVZ}.

Let $(X,\cb)$ be a measurable space, where $\cb$ is an algebra of
subsets $X$. A function $\m:\cb\to \bq_p$ is said to be a {\it
$p$-adic measure} if for any $A_1,\dots,A_n\subset\cb$ such that
$A_i\cap A_j=\emptyset$ ($i\neq j$) the equality holds
$$
\mu\bigg(\bigcup_{j=1}^{n} A_j\bigg)=\sum_{j=1}^{n}\mu(A_j).
$$
A $p$-adic measure is called a {\it probability measure} if
$\mu(X)=1$. For more detail information about $p$-adic measures we
refer to \cite{K3},\cite{KhN},\cite{Ro}.

Let $\Gamma^k = (V,L)$ be a semi-infinite Cayley tree of order
$k\geq 1$ with the root $x^0$ (whose each vertex has exactly $k+1$
edges, except for the root $x^0$, which has $k$ edges). Here $V$ is
the set of vertices and $L$ is the set of edges. The vertices $x$
and $y$ are called {\it nearest neighbors} and they are denoted by
$l=<x,y>$ if there exists an edge connecting them. A collection of
the pairs $<x,x_1>,\dots,<x_{d-1},y>$ is called a {\it path} from
the point $x$ to the point $y$. The distance $d(x,y), x,y\in V$, on
the Cayley tree, is the length of the shortest path from $x$ to $y$.
Let us set
$$ W_n=\{x\in V| d(x,x^0)=n\}, \ \ \
V_n=\bigcup_{m=1}^n W_m, \ \ L_n=\{l=<x,y>\in L | x,y\in V_n\}.
$$ The set of {\it direct successors} of $x$ is defined by
\begin{equation}\label{S(x)}
S(x)=\{y\in W_{n+1} :  d(x,y)=1 \}, \ \ x\in W_n.
\end{equation}
Observe that any vertex $x\neq x^0$ has $k$ direct successors and
$x^0$ has $k+1$.

 Now we are going to introduce a coordinate
structure in $\G^k$. Every vertex $x$ (except for $x_0$) of $\G^k$
has coordinates $(i_1,\dots,i_n)$, here $i_m\in\{1,\dots,k\}$,
$1\leq m\leq n$ and for the vertex $x_0$ we put $\emptyset$.
  Namely, the symbol $\emptyset$
constitutes level 0 and the sites $(i_1,\dots,i_n)$ form level $n$
of the lattice. In this notation for $x\in \G^k$,
$x=(i_1,\dots,i_n)$ we have
$$
S(x)=\{(x,i):\ 1\leq i\leq k\},
$$
here $(x,i)$ means that $(i_1,\dots,i_n,i)$.


Let us define on $\G^k$ a binary operation
$\circ:\G^k\times\G^k\to\G^k$ as follows: for any two elements
$x=(i_1,\dots,i_n)$ and $y=(j_1,\dots,j_m)$ put
\begin{equation}\label{binar1}
x\circ
y=(i_1,\dots,i_n)\circ(j_1,\dots,j_m)=(i_1,\dots,i_n,j_1,\dots,j_m)
\end{equation}
and
\begin{equation}\label{binar2}
x\circ x^0=x^0\circ x= (i_1,\dots,i_n)\circ(0)=(i_1,\dots,i_n).
\end{equation}

By means of the defined operation $\G^k$ becomes a noncommutative
semigroup with a unit. Using this semigroup structure one defines
translations $\tau_g:\G^k\to \G^k$, $g\in \G^k$ by
\begin{equation}\label{trans1}
\tau_g(x)=g\circ x.
\end{equation}
It is clear that $\tau_{(0)}=id$.

Similarly, by means of $\tau_g$ one can define translations
$\tilde\tau_g: L\to L$ of $L$. Namely,
$$
\tilde\tau_g(<x,y>)=<\tau_g(x),\tau_g(y)>.
$$

Let $G\subset \G^k$ be a sub-semigroup of $\G^k$ and $h:\G^k\to \br$
be a function defined on $\G^k$. We say that $h$ is {\it
$G$-periodic} if $h(\tau_g(x))=h(x)$ for all $g\in G$ and $x\in
\G^2$. Any $\G^k$-periodic function is called {\it translation
invariant}. Put
\begin{equation}\label{sub}
G_k=\{x\in \G^k: \ d(x,x^0)\equiv 0 (\textrm{mod}\ k) \}, \ \ k\geq
2
\end{equation}
One can check that $G_k$ is a sub-semigroup with a unit.

\section{The $p$-adic Potts model}

We consider the $p$-adic Potts model where spin takes values in the
set $\Phi=\{0,1,2,\cdots,\}$ and is assigned to the vertices of the
tree $\G^k=(V,L)$. A configuration $\s$ on $V$ is then defined as a
function $x\in V\to\s(x)\in\Phi$; in a similar manner one defines
configurations $\s_n$ and $\s^{(n)}$ on $V_n$ and $W_n$,
respectively. The set of all configurations on $V$ (resp. $V_n$,
$W_n$) coincides with $\Omega=\Phi^{V}$ (resp.
$\Omega_{V_n}=\Phi^{V_n},\ \ \Omega_{W_n}=\Phi^{W_n}$). One can see
that $\Om_{V_n}=\Om_{V_{n-1}}\times\Om_{W_n}$. Using this, for given
configurations $\s_{n-1}\in\Om_{V_{n-1}}$ and
$\s^{(n)}\in\Om_{W_{n}}$ we define their concatenations  by
$$\s_{n-1}\vee\s^{(n)}=\{\{\s_n(x),x\in
V_{n-1}\},\{\s^{(n)}(y),y\in W_n\}\}.$$ It is clear that
$\s_{n-1}\vee\s^{(n)}\in \Om_{V_n}$.

The Hamiltonian $H_n:\Om_{V_n}\to\bq_p$ of the $p$-adic countable
state Potts model has the form
\begin{equation}\label{Potts}
H_n(\s)=\sum_{<x,y>\in L_n}N_{xy}\delta_{\s(x),\s(y)},  \ \
n\in\mathbb{N},
\end{equation}
here $\s\in\Om_n$, $N_{xy}\in\bn$ ($x,y\in V$) and $\delta$ is the
Kronecker symbol.

We are going to construct $p$-adic quasi Gibbs measures for the
model.

Let us consider a function $\h: x\in V\to
\h_x=\{h_{i,x}\}_{i\in\Phi}\in\bq_p^{\Phi}$ of $x\in
V\setminus\{x^{(0)}\}$. Here $\bq_p^{\Phi}=\{\{x_i\}_{i\in \Phi}:
x_i\in\bq_p\}.$

 Fix a sequence $\l\in\bq_p^{\Phi}$ such that
\begin{equation}\label{L}
|\l(n)|_p\to 0 \ \ \textrm{as} \ \ n\to\infty.
\end{equation}
Such kind of sequences are called {\it weights}.

Given $n=1,2,\dots$ define a $p$-adic probability measure
$\m^{(n)}_\h$ on $\Om_{V_n}$ by
\begin{equation}\label{mu}
\mu^{(n)}_{\h}(\s_n)=\frac{1}{Z_n^{(\h)}}\frak{p}^{H_n(\s_n)}\prod_{x\in
W_n}h_{\s_n(x),x}\prod_{x\in V_n}\l(\s_n(x))
\end{equation}
where $\frak{p}$ is any $p$-adic number such that $|\frak{p}|_p\leq
1$. In particular, it could be $\frak{p}=p$. Here, as before,
$\s_n:x\in V_n\to\s_n(x)\in\Phi$ and $Z_n^{(\h)}$ is the
corresponding partition function:
\begin{equation}\label{Zn}
Z_n^{(\h)}=\sum_{\s\in\Omega_{V_n}}\frak{p}^{H_n(\s)}\prod_{x\in
W_n}h_{\s(x),x}\prod_{x\in V_n}\l(\s(x)).
\end{equation}

One of the central results of the theory of probability concerns a
construction of an infinite volume distribution with given
finite-dimensional distributions, which is called {\it Kolmogorov's
Theorem} \cite{Sh}. Therefore, in this paper we are interested in
the same question but in a $p$-adic context. More exactly, we want
to define a $p$-adic probability measure $\m$ on $\Om$ which is
compatible with defined ones $\m_\h^{(n)}$, i.e.
\begin{equation}\label{CM}
\m(\s\in\Om: \s|_{V_n}=\s_n)=\m^{(n)}_\h(\s_n), \ \ \ \textrm{for
all} \ \ \s_n\in\Om_{V_n}, \ n\in\bn.
\end{equation}

In general, \`{a} priori the existence such a kind of measure $\m$
is not known, since there is not much information on topological
properties, such as compactness, of the set of all $p$-adic measures
defined even on compact spaces \footnote{In the real case, when the
state space is compact, then the existence follows from the
compactness of the set of all probability measures (i.e. Prohorov's
Theorem). When the state space is non-compact, then there is a
Dobrushin's Theorem \cite{Dob1,Dob2} which gives a sufficient
condition for the existence of the Gibbs measure for a large class
of Hamiltonians.}. Note that certain properties of the set of
$p$-adic measures has been studied in \cite{kas2,kas3}, but those
properties are not enough to prove the existence of the limiting
measure. Therefore, at a moment, we can only use the $p$-adic
Kolmogorov extension Theorem (see \cite{GMR},\cite{KL}) which based
on so called {\it compatibility condition} for the measures
$\m_\h^{(n)}$, $n\geq 1$, i.e.
\begin{equation}\label{comp}
\sum_{\w\in\Om_{W_n}}\m^{(n)}_\h(\s_{n-1}\vee\w)=\m^{(n-1)}_\h(\s_{n-1}),
\end{equation}
for any $\s_{n-1}\in\Om_{V_{n-1}}$. This condition according to the
theorem implies the existence of a unique $p$-adic measure $\m$
defined on $\Om$ with a required condition \eqref{CM}. Note that
more general theory of $p$-adic measures has been developed in
\cite{kas1}.

So, if for some function $\h$ the measures  $\m_\h^{(n)}$ satisfy
the compatibility condition, then there is a unique $p$-adic
probability measure, which we denote by $\m_\h$, since it depends on
$\h$. Such a measure $\m_\h$ is said to be {\it a $p$-adic quasi
Gibbs measure} corresponding to the $p$-adic Potts model. By
$Q\cg(H)$ we denote the set of all $p$-adic quasi Gibbs measures
associated with functions $\h=\{\h_x,\ x\in V\}$. If there are at
least two distinct $p$-adic quasi Gibbs measures $\m,\n\in Q\cg(H)$
such that $\m$ is bounded and $\n$ is unbounded, then we say that
{\it a phase transition} occurs. By another words, one can find two
different functions $\sb$ and $\h$ defined on $\bn$ such that there
exist the corresponding measures $\m_\sb$ and $\m_\h$, for which one
is bounded, another one is unbounded. If there are two different
functions $\sb$ and $\h$ defined on $\bn$ such that there exist the
corresponding measures $\m_\sb$, $\m_\h$, and they are bounded, then
we say there is a {\it quasi phase transition}.

\begin{rem} Note that in \cite{KM} we considered the
following sequence of $p$-adic measures defined by
\begin{equation}\label{mr-mu}
\mu^{(n)}_{\h}(\s_n)=\frac{1}{\tilde
Z_n^{(\h)}}\exp_p\bigg\{H_n(\s_n)+\sum_{x\in
W_n}h_{\s_n(x),x}\bigg\}\prod_{x\in V_n}\l(\s_n(x))
\end{equation}
here as usual $\tilde Z_n^{(\h)}$ is the corresponding normalizing
factor. A limiting $p$-adic measures generated by \eqref{mr-mu} was
called {\it $p$-adic Gibbs measure}. Such kind of measures have been
studied in \cite{KM}.
\end{rem}

Now let us find for what kind of functions $\h=\{\h_x : x\in V\}$
the measures defined by \eqref{mu} would satisfy the compatibility
condition \eqref{comp}. The following statement describes conditions
on $\h$ guaranteeing the compatibility condition for the measures
$\m^{(n)}_\h$.

\begin{thm}\label{comp1} The measures $\m^{(n)}_\h$, $
n=1,2,...$ satisfy the compatibility condition \eqref{comp} if and
only if for any $x\in V\setminus\{x^{(0)}\}$ the following
equation holds:
\begin{equation}\label{eq1}
\hat h_{i,x}=\frac{\l(i)}{\l(0)}\prod_{y\in S(x)}F_i(\hat
\h_y;\theta_{xy}), \ \ i\in\bn
\end{equation}
here and below $\theta_{xy}=\frak{p}^{N_{xy}}$, a vector $\hat
\h=\{\hat h_i\}_{i\in\bn}\in\bq_p^\bn$ is defined by a vector
$\h=\{h_i\}_{i\in\Phi}$ as follows
\begin{equation}\label{H}
\hat h_i=\frac{h_i\l(i)}{h_0\l(0)}, \ \ \ i\in\bn
\end{equation}
and a mapping $F:\bq_p^{\bn}\to\bq_p^{\bn}$ is
$F(\xb;\theta)=\{F_i(\xb;\theta)\}_{i\in\bn}$ with
\begin{equation}\label{eq2}
F_i(\xb;\theta)=\frac{(\theta-1)x_i+\sum\limits_{j=1}^{\infty}x_j+1}
{\sum\limits_{j=1}^{\infty}x_j+\theta}, \ \ \xb=\{x_i\}_{i\in\bn}, \
\ i\in\bn.
\end{equation}
\end{thm}

The proof consists of checking the condition \eqref{comp} for the
measures \eqref{mu} (cp. \cite{GR},\cite{KM}).

\begin{rem}\label{r2} Note that thanks to non-Archimedeanity of the
norm $|\cdot|_p$ the series $\sum\limits_{k=1}^\infty x_k$ converges
iff the sequence $\{x_n\}$ converges to 0 (see \cite{Ko}).
Therefore, in what follows we should assume that
$\sum\limits_{k=1}^\infty \hat h_{k,x}$ converges, otherwise the
equation \eqref{eq1} has no sense.
\end{rem}

\begin{rem}\label{r3} In what follows, without loss of generality, we may
assume that $h_0=1$ and $\l(0)=1$. Otherwise, in \eqref{mu} we
multiply and divide the expression on the right hand side by
$\prod_{x\in W_n}h_{0,x}\prod_{x\in V_n}\l(0)$, and after replacing
$h_i$ by $h_i/h_0$ and $\l(k)$ by $\l(k)/\l(0)$, respectively, we get the desired equality.\\
\end{rem}

\begin{obs}\label{o1} Here we are going to underline a connection
between $q$-state Potts model with the defined one. First recall
that $q$-state Potts model is defined by the same Hamiltonian
\eqref{Potts}, but with the state space $\Phi_q=\{0,1,\dots,q-1\}$.
Similarly, one can define $p$-adic quasi Gibbs measures for the
$q$-state Potts model, here instead of the weight $\{\l(i)\}$ we
will take a collection $\{\l(0),\l(1),\dots,\l(q-1)\}\subset\bq_p$.

Now consider countable Potts model with a weight $\{\l(i)\}$ such
that
\begin{equation}\label{qq}
\l(k)=0 \ \ \textrm{for all} \ k\geq q, \ q>1.
\end{equation}
In this case the corresponding $p$-adic quasi Gibbs measures will
coincide with those of $q$-state Potts model. Indeed, let
\begin{eqnarray*}
&& \Om^c=\{\s\in\Om:\ \exists j\in\bz_+:\ \s(j)\geq q\}\\
&& \Om^{(q)}=\{\s\in\Om:\  \s(j)\leq q-1 \ \forall j\in\bz_+\}
\end{eqnarray*}
It is clear that $\Om^{(q)}=\Phi_q^{\bz_+}$. Let $\m$ be a $p$-adic
quasi Gibbs measure of the countable Potts model with the given
weight corresponding to a solution $\h_n=\{h_{i,n}\}_{i\in\Phi}$ of
\eqref{eq1}. From the definition \eqref{mu} we see that the
restriction of $\m$ to $\Om^c$ is zero, i.e. $\m\lceil_{\Om^c}=0$.
Moreover, from \eqref{eq1} and \eqref{qq} we conclude that
$h_{i,n}=0$ for all $i\geq q$. This means that vectors
$\h_n^{(q)}=\{h_{i,n}\}_{i\in\Phi_q}$ will be a solution of
\eqref{eq1} corresponding to the $q$-state Potts model. Therefore,
the restriction of $\m$ to $\Om^{(q)}$ coincides with $p$-adic quasi
Gibbs measure of $q$-state Potts model with a weight
$\{\l(0),\l(1),\dots,\l(q-1)\}$ corresponding to a solution of
$\h_n^{(q)}$.

Hence, we conclude that under the condition \eqref{qq} all $p$-adic
quasi Gibbs measures corresponding to countable Potts model are
described by those measures of $q$-state Potts model.
\end{obs}

\section{Existence of $p$-adic quasi Gibbs measure}

In this section we are going to provide a condition for the
existence of the $p$-adic quasi Gibbs measure.

Taking into account Remark \ref{r2}, we consider the following space
$$c_0=\{\{x_n\}_{n\in\bn}\subset\bq_p: \ |x_n|_p\to 0 \ \
\textrm{as} \ \ n\to\infty\}$$ with a norm
$\|x\|=\max\limits_n|x_n|_p$. Define
$$\Bb_r=\{\{x_n\}_{n\in\bn}\in
c_0: \ \|x\|\leq r\}, \ \ \ r>0.
$$
It is clear that $\Bb_r$ is a
closed set.

\begin{lem}\label{inv} Let $|\t|_p\geq 1$. Then
\begin{equation}\label{FF}
|F_i(\xb,\t)-F_i(\yb,\t)|_p\leq\|\xb-\yb\| \ \ \textrm{for every} \
\ \xb,\yb\in\Bb_\d.
\end{equation}
\end{lem}

\begin{proof} Let $\xb=(x_i),\yb=(y_i)\in\Bb_\d$, then it is clear that $|x_i|_p\leq \d$ and
$|y_i|_p\leq \d$ for all $i\in\bn$. Therefore, the strong triangle
inequality implies that
\begin{equation}\label{denom}
\bigg|\sum_{j=1}^\infty x_j+\t\bigg|_p=|\t|_p, \ \
\bigg|\sum_{j=1}^\infty y_j+\t\bigg|_p=|\t|_p.
\end{equation}

From \eqref{eq2} one gets
\begin{eqnarray*}
|F_i(\xb,\t)-F_i(\yb,\t)|_p&=&\frac{|\t-1|_p}{|\t|_p^2}\bigg|\bigg(\sum_{j=1}^\infty x_j+\t\bigg)(x_i-y_i)+(1-x_i)\sum_{j=1}^\infty(x_j-y_j)\bigg|_p\nonumber\\
&\leq&\frac{1}{|\t|_p}\max\bigg\{|\t|_p|x_i-y_i|_p,|1-x_i|_p\bigg|\sum_{j=1}^\infty(x_j-y_j)\bigg|_p\bigg\}\nonumber\\
&\leq&\max_i|x_i-y_i|_p\nonumber \\
&=&\|\xb-\yb\|.
\end{eqnarray*}
This completes the proof.
\end{proof}

Before going to equation \eqref{eq1}, let us first enumerate $S(x)$
for any $x\in V$ as follows $S(x)=\{x_1,\cdots,x_k\}$, here as
before $S(x)$ is the set of direct successors of $x$ (see
\eqref{S(x)}). Using this enumeration one can rewrite \eqref{eq1} by
\begin{equation}\label{eq11}
\hat h_{i,x}=\l(i)\prod_{m=1}^k F_i(\hat\h_{x_m},\t_i), \ \ i\in\bn,
\ \ \textrm{for every} \ \ x\in V\setminus\{x^{(0)}\}.
\end{equation}

Now we can formulate the main result.

\begin{thm}\label{uniq} Assume that $N_{xy}\leq 0$ for all $<x,y>\in\L$, i.e.
$|\t_{xy}|_p\geq 1$ and for $\l$ the condition
\begin{equation}\label{Bk}
|\l(i)|_p\leq \d \qquad \forall i\in\bn,
\end{equation}
is satisfied. Then any two solution of \eqref{eq11} belonging to
$\Bb_\d$ coincides with each other.
\end{thm}

\begin{proof} Assume that $\hat\h=\{\h_x,x\in V\setminus\{x^0\}\}, \hat\sb=\{\hat\sb_x,x\in V\setminus\{x^0\}\}$
be two solutions of \eqref{eq11}. Now we are going to show that they
coincide with each other. Indeed, let $x\in V\setminus\{x^{(0)}\}$
be an arbitrary vertex. Using Lemma \ref{pr} from \eqref{eq11} one
can find
\begin{eqnarray*}
|\hat h_{i,x}-\hat s_{i,x}|_p\leq |\l(i)|_p\max_{1\leq m\leq
k}\bigg\{|F_i(\hat\h_{x_m}))_i-F_i(\hat\sb_{x_m})|_p\bigg\},
\end{eqnarray*}
which with \eqref{FF} and \eqref{Bk} implies that
\begin{equation}\label{hx}
\|\hat\h_x-\hat\sb_x\|\leq \d\max_{1\leq m\leq
k}\bigg\{\|\hat\h_{x_m}-\hat\sb_{x_m}\|\bigg\}
\end{equation}

Now take an arbitrary $\e>0$. Let $n_0\in\bn$ such that
$1/p^{n_0}<\e$. Iterating \eqref{hx} $n_0$ times one gets
$\|\hat\h_x-\hat\sb_x\|<\e$.

Hence, from the arbitrariness of $\e$ we obtain $\hat\h_x=\hat\sb_x$
for every $x\in V\setminus\{x^{(0)}\}$.
 This completes the proof.
\end{proof}

The provided theorem states, if equation \eqref{eq11} has solution
belonging to $\Bb_\d$, then it is unique. But, in general, we do not
whether the equation has a solution or not. Below, we provide a
sufficient condition for the existence of solution of \eqref{eq11}.

\begin{rem} The proved Theorem extends the main results of \cite{KM1,M}
to more general kind of measures, since there it was taken
$\frak{p}=\exp_p$.
\end{rem}

\subsection{Homogeneous case} In this subsection, we shall assume that $N_{xy}:=N$ for all
$<x,y>\in\L$. Namely, we want to consider homogeneous case.

Recall that a function $\h=\{\h_x\}_{x\in V\setminus\{x^0\}}$ is
called {\it translation-invariant} if $\h_{\tau_x(y)}=\h_y$ for all
$x,y\in V\setminus\{x^0\}$.  Let us restrict ourselves to the
description of translation-invariant solutions of \eqref{eq1},
namely $\h_x=\h(=(h_0,h_1,\dots,))$ for all $x\in V$.

Define the following mapping
\begin{equation}\label{F}
(\cf_\t(\xb))_i=\l(i)(F_i(\xb,\t))^k, \ \ \ i\in\bn,
\end{equation}
where $\xb=\{x_n\}\in c_0$. Here, $\t=\frak{p}^N$.  Due to Remark
\ref{r2}, the mapping $\cf_\t$ is well defined.

\begin{thm}\label{inv1} Let $N\leq 0$ i.e. $|\t|_p\geq 1$. Assume that for $\l$ the condition
\eqref{Bk} is satisfied. Then $\cf_\t(\Bb_\d)\subset \Bb_\d$, and
\begin{equation}\label{fxy}
\|\cf_\t(\xb)-\cf_\t(\yb)\|\leq\d\|\xb-\yb\| \ \ \textrm{for every}
\ \ \xb,\yb\in\Bb_\d.
\end{equation}
\end{thm}

\begin{proof} Let $\xb=(x_i),\yb=(y_i)\in\Bb_\d$. Then from \eqref{eq2} with \eqref{denom} we have
\begin{equation}\label{eq3}
|F_i(\xb,\t)|_p=\frac{\bigg|\t x_i+\sum\limits_{j=1, j\neq i}^\infty
x_j+1\bigg|_p}{|\t|_p}
\leq\frac{\max\{|\t|_p|x_i|_p,1\}}{|\t|_p}\leq 1
\end{equation}
this with \eqref{F} yields that $\cf_\t(\Bb_\d)\subset \Bb_\d$.

Now using \eqref{FF},\eqref{F},\eqref{eq3} we obtain

\begin{eqnarray*}
\|\cf_\t(\xb)-\cf_\t(\yb)\|&=&\max_i|\l(i)|_p\max_i|(F_i(\xb,\t))^k-
(F_i(\yb,\t))^k|_p\\
&\leq &\d|F_i(\xb,\t)-
F_i(\yb,\t)|_p\max_i\bigg|\sum_{\ell=0}^{k-1}F^{k-\ell}_i(\xb,\t)F^{\ell}_i(\yb,\t)\bigg|_p\\
&\leq &\d\|\xb-\yb\|.
\end{eqnarray*}
This completes the proof.
\end{proof}

\begin{rem} We should stress that if $N>0$ then the similar methods
are no longer applicable for $\cf_\t$, therefore, it needs other
kind of techniques. So, such a case will be considered elsewhere.
\end{rem}

 Now thanks to  Lemma \ref{inv1} we can apply the fixed point
theorem to $\cf_\t$, which implies the existence of  unique fixed
point $\xb_0\in\Bb_\d$. This, according to Theorem \ref{uniq}, means
that there exists a unique solution of \eqref{eq1}. Hence, due to
Theorem \ref{comp1}, such a solution defines the $p$-adic quasi
Gibbs measure $\m_0$.

\begin{rem}\label{r5}  Let us emphasize the following notes:
\begin{enumerate}
\item[(a)]  Note that  in \cite{Mq} we have proved for the $q+1$-state Potts
model the $p$-adic quasi Gibbs measure is unique if $q$ and $p$ are
relatively prime. Therefore, the proved Theorem \ref{uniq} shows the
difference between finite and countable state Potts models.
Moreover, in the real case such kind of result is unknown (see
\cite{Ga,GR}).

\item[(b)] We also should stress that the condition \eqref{Bk} is
important. If we replace $\d$ with 1, then Theorem \ref{uniq} may
not be valid. Namely, in that case it may occur a quasi phase
transition. Indeed, if $\l(1)=\l(2)=1$ and $\l(k)=0$ for every
$k\geq 3$, then clearly \eqref{Bk} is not satisfied. On the other
hand, our model reduces to 3-state Potts model. For such a model in
\cite{Mq} the existence of the quasi phase transition has been
proved at $p=2$.
\end{enumerate}
\end{rem}

\subsection{Periodic case} In this subsection, we consider when $N_{xy}$ is $G_m$-periodic.
This means that $N_{\tau_g(x)\tau_g(y)}=N_{xy}$ for all $g\in G_m$
and $x,y\in \G^m$. Therefore, let us denote
$$
\t_i=\frak{p}^{N_{xy}}, \ \ \textrm{if}\ \ d(x,x^0)\equiv
i(\textrm{mod}\ m), \ i=0,\dots,m-1.
$$

We want to find $G_m$-periodic solution of \eqref{eq1}. Recall that
a function $\h=\{\h_x\}_{x\in V\setminus\{x^0\}}$ is called {\it
$G_m$-periodic} if $\h_{\tau_g(x)}=\h_x$ for all $g\in G_m$, $x\in
V\setminus\{x^0\}$. Note that if $\h=\{\h_x\}$ is a $G_k$-periodic,
then it can be defined by a $k$-collection of vectors
$(\h_0,\dots,\h_{k-1})$, where $\h_i=\{h^{(i)}_j\}_{j\in\bn}$, i.e.
$\h_x=\h_i$, if $d(x,x^0)\equiv i(\textrm{mod}\ m)$,
$i=0,\dots,k-1$.

Then equation \eqref{eq1} is reduced to the following one
\begin{eqnarray}\label{k-eq1}
\left\{
\begin{array}{ll}
h^{(i)}_j=\l_j\big(F_j(\h^{(i+1)},\t_{i+1})\big)^k, \ \ i=0,\dots,
m-2\\
h^{(m-1)}_j=\l_j\big(F_j(\h^{(0)},\t_{0})\big)^k,
\end{array}
\right. \ \ j\in\bn
\end{eqnarray}

Hence, define the following mapping
\begin{equation}\label{H}
(\ch(\xb))=\cf_{\t_{m-1}}(\cf_{\t_{m-2}}(\cdots(\cf_{\t_{0}}(\xb))\cdots))
\end{equation}
where $\xb=\{x_n\}\in c_0$.

It is clear that the fixed point of \eqref{H} defines a solution of
\eqref{k-eq1}.

\begin{thm}\label{invh} Let $|\t_i|_p\geq 1$, $i=0,\dots,m-1$. Assume that for $\l$ the condition
\eqref{Bk} is satisfied. Then $\ch(\Bb_\d)\subset \Bb_\d$, and
\begin{equation}\label{fxy}
\|\ch(\xb)-\ch(\yb)\|\leq\d^m\|\xb-\yb\| \ \ \textrm{for every} \ \
\xb,\yb\in\Bb_\d.
\end{equation}
\end{thm}

The proof immediately follows from Theorem \ref{inv1}.

Consequently, Theorem \ref{invh} allows us to apply the fixed point
theorem to $\ch$, which implies the existence of  unique fixed point
$\xb_0\in\Bb_\d$. This, according to Theorem \ref{uniq}, means that
there exists a unique solution of \eqref{eq1}. Hence, due to Theorem
\ref{comp1}, such a solution defines the $p$-adic quasi Gibbs
measure $\m_m$.

\section{Conclusion}

In the present study we have provided a new construction of measure,
called $p$-adic quasi Gibbs measure, for countable state $p$-adic
Potts model on the Cayley tree. Note that the construction depends
on a parameter $\frak{p}$ and wights $\{\l(i)\}_{i\in\bn}$. In
particular case, i.e. if $\frak{p}=\exp_p$, the defined measure
coincides with $p$-adic Gibbs measure (see \cite{KM,KM1}). In this
paper, under some condition on weights we proved the existence of
$p$-adic quasi Gibbs measures associated with the model. Note that
this condition does not depend on values of the prime $p$. Moreover,
our result extends the previous proved ones in \cite{KM1,M}. An
analogues fact is not valid when the number of spins is finite. We
should stress that when states are finite and $\rho=p$, then the
corresponding $p$-adic quasi Gibbs measures have been investigated
in \cite{Mq}.

\section*{Acknowledgement} The present study have been done within
the IIUM grant EDW B11-159-0637. The author also acknowledges the
Junior Associate scheme of the Abdus Salam International Centre for
Theoretical Physics, Trieste, Italy.


\begin{thebibliography}{99}


\bibitem{ADFV}  I. Ya. Aref\'eva, B. Dragovic, P.H. Frampton, I.V. Volovich, The wave function
of the Universe and $p -$ adic gravity, \textit{Int. J. Modern Phys.
A} {\bf 6}(24) (1991), 4341--4358.

\bibitem{ABK}  V.A. Avetisov, A.H. Bikulov, S.V. Kozyrev, Application of
pЦadic analysis to models of spontaneous breaking of the replica
symmetry, \textit{ J. Phys. A: Math. Gen.}, {\bf  32}(1999),
8785--8791.

\bibitem{BC}  E. Beltrametti, G. Cassinelli, Quantum mechanics and $p-$ adic numbers,
\textit{Found. Phys.} {\bf 2} (1972), 1--7.

\bibitem{DF} M. Del Muto, A. Fig$\grave{a}$-Talamanca, Diffusion on locally
compact ultrametric spaces, \textit{Expo. Math.} {\bf 22} (2004),
197Ц-211.


\bibitem{Dob1} R.L. Dobrushin, The problem of uniqueness of a
Gibbsian random field and the problem of phase transitions.
\textit{Funct.Anal. Appl.} {\bf 2}(1968)2, 302--312.

\bibitem{Dob2} R.L. Dobrushin, Prescribing a system of random variables by conditional distributions.
\textit{Theor. Probab. Appl.} {\bf 15}(1970), 458--486.

\bibitem{FW}  P.G.O. Freund, E. Witten,
Adelic string amplitudes, \textit{Phys. Lett. B} {\bf 199} (1987),
191--195.


\bibitem{Ga} N.N. Ganikhodjaev, The Potts model on $\bz^d$  with countable set of spin values,
\textit{Jour. Math. Phys.} {\bf 45} (2004), 1121--1127.

\bibitem{GMR} N.N. Ganikhodjaev, F.M. Mukhamedov, U.A. Rozikov,
Phase transitions of the Ising model on $\bz$ in the $p$-adic number
field. \textit{Uzbek. Mat. Jour.} (1998), No. 4 , 23--29 (Russian).


\bibitem{GR} N.N. Ganikhodjaev, U.A. Rozikov, The Potts model with countable set
of spin values on a Cayley tree, \textit{Lett. Math. Phys.} {\bf 75}
(2006), 99--109.

\bibitem{G} H.O. Georgii, \textit{ Gibbs measures and phase
transitions}, (Walter de Gruyter, Berlin, 1988).


\bibitem{kas1}  A.K. Katsaras,  Extensions of $p$-adic vector measures,
\textit{Indag. Math.N.S.} {\bf 19} (2008) 579--600.

\bibitem{kas2}  A.K. Katsaras,  On spaces of $p$-adic vector measures,
    \textit{P-Adic Numbers, Ultrametric Analysis, Appl.} {\bf  1} (2009)  190--203.

\bibitem{kas3} A.K.Katsaras,  On $p$-adic vector measures,
    \textit{Jour. Math. Anal. Appl.} {\bf 365} (2010), 342Ц-357.

\bibitem{K3} A.Yu.Khrennikov,  $p$-adic valued probability measures,
\textit{Indag. Mathem. N.S.} {\bf 7}(1996), 311-330.

\bibitem{Kh1} A.Yu.Khrennikov,  \textit{$p$-adic Valued Distributions in Mathematical
Physics}, Kluwer Academic Publisher, Dordrecht, 1994.

\bibitem{Kh2} A.Yu.Khrennikov, \textit{ Non-Archimedean analysis: quantum paradoxes,
dynamical systems and biological models}, Kluwer Academic Publisher,
Dordrecht, 1997.

\bibitem{K4} A.Yu.Khrennikov,  Limit behaviour of sums of independent random variables with
respect to the uniform p-adic distribution, \textit{Statis. \&
Probab. Lett.} {\bf 51}(2001), 269--276.

\bibitem{KK1} A.Yu.Khrennikov, S.V.Kozyrev, Wavelets on ultrametric spaces,
\textit{Appl. Comput. Harmonic Anal.}, {\bf 19}(2005) 61--76.

\bibitem{KK2} A.Yu. Khrennikov, S.V. Kozyrev, Ultrametric random field.
Infin. Dimens. Anal. Quantum Probab. Relat. Top. {\bf 9}(2006),
199-213.

\bibitem{KK3} A.Yu. Khrennikov, S.V. Kozyrev,  Replica symmetry breaking
related to a general ultrametric space I,II,III, \textit{Physica A},
{\bf 359}(2006), 222-240; 241-266; {\bf 378}(2007), 283-298.

\bibitem{KL} A. Khrennikov, S. Ludkovsky,  Stochastic processes on
non-Archimedean spaces with values in non-Archimedean fields.
\textit{Markov Process. Related Fields} {\bf 9} (2003), 131--162.

\bibitem{KhN} A.Yu. Khrennikov, M. Nilsson,  \textit{$p$-adic deterministic and
random dynamical systems}, Kluwer, Dordreht, 2004.

\bibitem{KM} A. Khrennikov, F. Mukhamedov, J.F.F. Mendes, On $p$-adic Gibbs measures of countable
state Potts model on the Cayley tree, \textit{Nonlinearity} {\bf
20}(2007) 2923Ц-2937.

\bibitem{KM1} A. Khrennikov, F. Mukhamedov, On uniqueness of Gibbs measure for p-adic countable state Potts model
on the Cayley tree, \textit{Nonlin. Anal.: Theor. Methods Appl.}
{\bf 71} (2009), 5327-Ц5331.

\bibitem{KYR} A.Yu.Khrennikov, S Yamada, A. van Rooij, Measure-theoretical
approach to $p$-adic probability theory, \textit{Annals Math. Blaise
Pascal}, {\bf 6}(1999) 21-32.

\bibitem{Kh07} A.Yu. Khrennikov, Generalized Probabilities
Taking Values in Non-Archimedean Fields and in Topological Groups,
\textit{Russian J. Math. Phys.} {\bf 14} (2007), 142-Ц159.

\bibitem{Ko} N. Koblitz,  \textit{$p$-adic numbers, $p$-adic analysis and
zeta-function}, Berlin, Springer, 1977.

\bibitem{Lu} S.V. Ludkovsky, Non-Archimedean valued quasi-invariant
descending at infinity measures.  \textit{Int. J. Math. Math. Sci.}
(2005), no. 23, 3799--3817.

\bibitem{MP} E.Marinary, G.Parisi,  On the $p$-adic five point function, \textit{Phys.
Lett. B} {\bf 203}(1988), 52-56.


\bibitem{MR1} F.M.Mukhamedov, U.A.Rozikov, On Gibbs measures of $p$-adic
Potts model on the Cayley tree, \textit{Indag. Math. N.S.} {\bf
15}(2004), 85--100.

\bibitem{MR2} F.M.Mukhamedov, U.A.Rozikov, On inhomogeneous $p$-adic Potts model on a Cayley tree,
\textit{Infin. Dimens. Anal. Quantum Probab. Relat. Top.} {\bf 8}
(2005), 277--290.

\bibitem{M} F. Mukhamedov, On existence of generalized Gibbs measures for one dimensional
$p$-adic countable state Potts model, \textit{ Proc. Steklov Inst.
Math.} {\bf 265} (2009), 165Ц-176.

\bibitem{Mq}  F.M.Mukhamedov, On $p$-adic quasi Gibbs measures for $q+1$-state Potts
model on Cayley tree, \textit{P-adic Numbers, Ultametric Anal.
Appl.} {\bf 2}(2010), 241--251.

\bibitem{Ro} A. van Rooij, \textit{Non-archimedean functional analysis}, Marcel Dekker, New York, 1978.

\bibitem{Sh} A.N. Shiryaev, \textit{Probability}, Nauka, Moscow, 1980.

\bibitem{VVZ} V.S.Vladimirov, I.V.Volovich, E.I.Zelenov, \textit{$p$-adic Analysis and
Mathematical Physics}, World Scientific, Singapour, 1994.

\bibitem{W} F.Y. Wu, The Potts model, \textit{Rev. Mod. Phys.}
{\bf 54} (1982), 235--268.


\end{thebibliography}
\end{document}